\documentclass[12pt, onecolumn]{IEEEtran}
\usepackage{subfigure}
\usepackage{epsf}
\usepackage{amsmath,amssymb}
\usepackage{graphicx}
\usepackage{color}
\usepackage{cite}
\usepackage{multirow,tabularx}
\usepackage{ifthen}
\usepackage{times}

\usepackage{stackrel}

\usepackage{amsfonts}
\usepackage{stfloats}

\usepackage[all]{xy}

\input{epsf.sty}

\newcommand{\hide}[1]{\ifthenelse{\boolean{false}}{#1}{}}

\newtheorem{thm}{Theorem}
\newtheorem{lemma}{{\bf Lemma}}
\newtheorem{rem}{Remark}

\newcommand{\dee}{\mathrm{d}}
\newcommand{\yell}{\mathcal{L}}
\newcommand{\ellm}{l_m}

\newcommand{\qed}{\nobreak \ifvmode \relax \else
      \ifdim\lastskip<1.5em \hskip-\lastskip
      \hskip1.5em plus0em minus0.5em \fi \nobreak
      \vrule height0.75em width0.5em depth0.25em\fi}

\IEEEoverridecommandlockouts
\begin{document}
%
\title{On Whitespace Identification Using Randomly Deployed Sensors}
\author{Rahul Vaze\IEEEauthorrefmark{1} and Chandra R. Murthy\IEEEauthorrefmark{2}
\thanks{\IEEEauthorrefmark{1}R. Vaze is with the School of Technology and Computer Science, Tata Institute of Fundamental Research, Homi Bhabha Road, Mumbai 400 005, India. \IEEEauthorrefmark{2}C. R. Murthy is with the Dept.\ of Electrical Communication
     Engg.\ at IISc, Bangalore,
     India. (e-mails: vaze@tcs.tifr.res.in, cmurthy@ece.iisc.ernet.in.)}

   \thanks{This work was partially supported by a research project funded by
     the Aerospace Network Research Consortium. }
     
}
%

\maketitle
\vspace{-1cm}
\begin{abstract} 
This work considers the identification of the available whitespace, i.e., the regions that are not covered by any of the existing transmitters,  within a given geographical area. To this end, $n$ sensors are deployed at random locations within the area. These sensors detect for the presence of a transmitter within their radio range $r_s$, and their individual decisions are combined to estimate the available whitespace. The limiting behavior of the recovered whitespace as a function of $n$ and $r_s$ is analyzed. It is shown that both the fraction of the available whitespace that the nodes fail to recover as well as their radio range both optimally scale as $\log(n)/n$ as $n$ gets large. The analysis is extended to the case of unreliable sensors, and it is shown that, surprisingly, the optimal scaling is still $\log(n)/n$ even in this case. A related problem of estimating the number of transmitters and their locations is also analyzed, with the sum absolute error in localization as performance metric. The optimal scaling of the radio range and the necessary minimum transmitter separation is determined, that ensure that the sum absolute error in transmitter localization is minimized, with high probability, as $n$ gets large. Finally, the optimal distribution of sensor deployment is determined,  given the distribution of the transmitters, and the resulting performance benefit is characterized.  
\end{abstract} 
%
\section{Introduction}
Whitespace identification, or determining the regions within a given geographical area of interest that are not covered by any of the existing transmitters, is useful in several applications~\cite{sandvig2006, feki2008, Mateos2009}. For service providers, it is important for finding  \emph{dead zones}, i.e., the coverage holes, in their service area. For cognitive radio (CR) networks, knowledge of the available whitespace is crucial in order to ensure that the CR nodes do not cause harmful interference to the primary transmitters. This paper addresses this problem using an approach where $n$ sensors are deployed in the geographical area of interest. Their collective observations regarding the presence of a transmitter in their \emph{radio range} $r_s$ are used to estimate the available whitespace. A  challenge here is to determine the optimal scaling of the radio range and the minimum error in the estimated whitespace, as a function of the number of sensors deployed. 

In the recent literature, various approaches for whitespace identification have been considered. One approach uses the receive signal strength (RSS) measurements obtained from the sensors to estimate the number of transmitters and their powers by minimizing the sum of the MSE in the location and power estimates~\cite{nasif2009measurements}. 
The problem of localizing a transmitter using binary observations instead of using analog RSS measurements has also been considered~\cite{niu2006, shoari2010localization, YR2012}. 

A related problem that has received considerable research attention in recent years is that of tracking one or more targets with the help of  multiple sensors (see \cite{Madhow2009} for comprehensive  review of the literature). More specifically, the problem of target tracking with binary sensor measurements has been studied both theoretically \cite{Madhow2009, Madhow2008WS, Aslam2003WS, Manjunath2012WS} and experimentally \cite{Kim2005WS}. In \cite{Madhow2009}, one or more targets that are located arbitrarily in the field of interest are tracked using multiple sensors, while in \cite{Madhow2008WS, Manjunath2012WS}, binary sensor measurements are used to decide on the presence or absence of targets at given locations. However, to the best of our knowledge, there have been no studies in the literature on the limiting behavior of whitespace recovery methods as the number of sensors deployed is increased. Of particular interest are questions related to the optimal scaling of the radio range of the sensors to achieve the minimum whitespace recovery error,  the resulting recovery performance, the optimum spatial distribution of sensors, etc. 

In this paper, we consider a scenario where $n$ sensors are deployed in a given geographical area for whitespace identification. For simplicity, and without  loss of generality, we consider the unit interval $\yell \triangleq [0,1]$ as the area of interest. The $n$ sensors are deployed uniformly at random locations in $\yell$. These sensors detect the presence of a transmitter within a distance $r_s$ from their location, and return a $1$ if there is at least one transmitter in their vicinity, and return $0$ otherwise. The total whitespace recovered is the union of the $2r_s$-length areas around the sensors that return $0$. Our contributions are:
\begin{enumerate}
\item  We show that both the whitespace recovery error (loss), i.e., the fraction of the available whitespace that is not recovered by the $n$ sensors, and the radio range $r_s$ both optimally scale as $\log(n)/n$ as $n$ gets large.\footnote{All logarithms in the sequel are to the base $e$.} 
\item  We extend the analysis to the case where the sensors report erroneous measurements with probability $p$. Surprisingly, even with unreliable sensors, the optimal scaling of the whitespace recovery error and the optimal radio range is still $\log(n)/n$. 
\item We also consider the problem of  transmitter localization, i.e., that of determining the number of active transmitters and their locations. With the sum absolute error in transmitter localization as the metric, we derive the optimum radio range as well as the minimum separation between transmitters that guarantees that the localization error is below a threshold with high probability. 
\item For a given spatial distribution of the transmitters, we derive closed-form expressions for the optimal spatial distribution of sensors that minimizes the probability of not detecting a transmitter and the resulting minimum miss-detection probability. 
\end{enumerate}
We validate our analytical results through Monte Carlo simulations. 
The simulation results also illustrate the significant performance improvement that is obtainable by using the optimal scaling for $r_s$, as the number of sensors $n$ is increased, compared to using a slower or faster decrease of $r_s$ with $n$~(see Fig.~\ref{fig:estnumtx}). Moreover, even though the results are true for large $n$, the scaling of $r_s = \log(n)/n$ is \emph{optimal} even at moderate or low values of $n$. 

Under a similar binary observation model in a 2-dimensional region with fixed sensor placements, it has been shown in \cite{Madhow2009} that the {\it expected} whitespace identification error scales as $\frac{1}{\rho r_s}$, where $\rho$ is the density of sensors and $r_s$ is the radio range. In \cite{Madhow2009}, however, $r_s$ is assumed to remain fixed as $\rho$ is increased, i.e., the results do not hold if $r_s$ is allowed to decrease as $\rho$ increases. In this paper, with $\rho = n$, we are interested in optimal scaling of $r_s$ with $n$, as $\rho$ is increased. Further, in our model, sensors are placed randomly in the region of interest and we obtain results that hold \emph{with high probability}. 
We show that the optimal radio range scales as $r_s=\Theta\left(\frac{\log n}{n}\right)$, 
and we get a localization error of $\Theta\left(\frac{\log n}{n}\right)$ with high probability. 

All of our results directly extend to 2-dimensional regions, with the quantities of interest such as the optimal radio range, optimal transmitter localization error, etc.\ being the square-roots of their counterparts in the one-dimensional case.  
Our results yield useful insights into the number of sensors to be deployed and their radio range for detecting transmitters that maximizes the recovered whitespace and accurately localizes the transmitters within a given geographical area. 

The organization of this paper is as follows. In Sec.~\ref{sec:sysmodel}, we present the modeling assumptions and problem setup. In Sec.~\ref{sec:whitespace1}, we derive analytical results on the whitespace identification  when the sensors are perfectly reliable. Section~\ref{sec:unreliable} extends the results to the case of unreliable sensors, and Sec.~\ref{sec:localization} extends the results to jointly identifying the number of transmitters and their locations, with the sum absolute error in localization as the metric. Section~\ref{sec:distribution} presents the optimum distribution of sensors that minimizes the probability of missing a transmitter. Simulation results are presented in Sec.~\ref{sec:sims}, and concluding remarks are offered in Sec.~\ref{sec:conc}.

\section{System Model} \label{sec:sysmodel}
We consider a unit length segment $\yell$, wherein $M$ transmitters\footnote{In the sequel, we will interchangeably use the phrases ``transmitters'' and ``primary transmitters'' to refer to the transmitters whose locations and transmission footprints we wish to determine.} are arbitrarily located. 
We assume that $n$ sensor are thrown uniformly at random locations on $\yell$. Each sensor has radio range $r_s(n)$, i.e., it can detect the presence of any transmitter that is at most $r_s(n)$ distance away. 
Each sensor returns one of two possible readings $b$, $b = 1$ if there is at least one transmitter at a distance of $r_s(n)$ from it, and $b = 0$ otherwise. The sensor readings are combined at a fusion center to find the region $\mathcal{A}_{\text{void}}$ of $\yell$ that is guaranteed not to contain any transmitter.  Now, if $x_1, x_1, \ldots, x_n$ are the sensor locations in $\yell$ and $b_1, b_2, \ldots, b_n$ are the corresponding sensor readings, then 
\begin{equation}
\mathcal{A}_{\text{void}} = \bigcup_{i=1}^n (1-b_i)  [\min(x_i - r_s, 0), \max(x_i + r_s, 1)].
\end{equation}
Let $\ell(\mathcal{A}) = \int_{x \in \mathcal{A}} \mathrm{d}x$ denote the length of a region denoted by $\mathcal{A}$. For example, if $\mathcal{A}$ is the union of a finite set of disjoint regions, then $\ell(\mathcal{A})$ is the sum of the lengths of the disjoint regions. We define $A_{\text{void}} = \ell(\mathcal{A}_{\text{void}})$ as the length of the region  where no transmitter is located. Note that, since the transmitters are located at distinct points that occupy no area, we would expect that, as $n\rightarrow \infty$, the transmitters are perfectly localized, and, $A_{\text{void}} \rightarrow 1$. Hence, formally speaking, we want to find the minimum $\epsilon(n)$ and a corresponding radio range $r_s(n)$ that guarantees that 
\begin{equation} \label{eq:prob}
\min_{r_s(n), \epsilon(n)}\lim_{n\rightarrow \infty} P\left( (1-A_{\text{void}}) \le \epsilon(n)\right) =1.
\end{equation}
The probability in the above equation is evaluated over the distribution of the sensor locations, where the transmitter locations are assumed to be fixed but arbitrary. This metric essentially captures the scaling of the  relative loss in recovering the whitespace, with increasing number of sensors, as a function of radio range $r_s(n)$. So, there are two problems to solve, i) finding the minimum scaling of the error $\epsilon(n)$ with $n$, and ii) finding the optimal radio range $r_s(n)$ as a function of~$n$.  


We note that, in addition to solving (\ref{eq:prob}), we may also wish to find  the number of transmitters and their locations.
Specifically, given an estimate of the number of transmitters and their locations, if we assume that each transmitter has a transmission range of $r_p$, then the available whitespace consists of the region of $\yell$ from which the subsets of size $2 r_p$ around each transmitter have been removed. We discuss how to determine the number of  transmitters and their locations in Sec.~\ref{sec:localization}.  The quantities $r_s$ and $r_p$ may be different, as the sensitivity of the sensor may be different from the sensitivity of the receiver to which the transmitter's signal was intended. 

The next section presents fundamental bounds on the whitespace recovery error and the corresponding optimal radio range, asymptotically in $n$, when the sensors are perfectly reliable. 



\section{Reliable Sensors} \label{sec:whitespace1}
We first present a lower bound on the error $\epsilon^{*}(n)$, that guarantees that if $\lim_{n\rightarrow \infty} \frac{\epsilon(n)}{\epsilon^{*}(n)} \rightarrow 0$, then  $\lim_{n\rightarrow \infty}P\left((1-A_{\text{void}}) \le \epsilon(n)\right)  \le c$, where $c < 1$ is a constant independent of $n$. In other words, we show that it is not possible to recover the available whitespace with an error smaller than $\epsilon^{*}(n)$ with arbitrarily high probability as $n$ gets large. To derive the lower bound, we will use the following result from the $1$-coverage problem in one dimension~\cite{kumar04kcoverage}. 

\begin{lemma}\label{lem:1-cov} Let $n$ sensors be thrown uniformly at random locations on the one-dimensional unit-length segment $\yell$, where each sensor has radio range of $r(n)$. A point $x$ on  $\yell$ is defined to be {\it covered} if there is at least one sensor in the interval $[x-r(n), \ x+r(n)]$. Then, if $\lim_{n\rightarrow \infty} \frac{r(n)}{\frac{\log n}{n}} \rightarrow 0$, then $\lim_{n\rightarrow \infty}P(\text{all points in $\yell$ are covered} ) < c_2$, where $c_2 <1$ is a constant. A similar result holds in 2-dimensions, with $\frac{r(n)}{\sqrt{\frac{\log n}{n}}}\rightarrow 0$, where a point is said to be covered if there is a sensor in a radius $r(n)$ around it.
\end{lemma}

\begin{thm}\label{thm:1Dlb} For the whitespace recovery problem in a one-dimensional unit-length region, if $\lim_{n\rightarrow \infty} \frac{\epsilon(n)}{\frac{\log n}{n}} \rightarrow 0$ or $\lim_{n\rightarrow \infty} \frac{r_s(n)}{\frac{\log n}{n}} \rightarrow 0$, then $P\left((1-A_{\text{void}}) \le \epsilon(n)\right)  < 1$.
\end{thm}

\begin{proof} Consider the case of a single transmitter, i.e., $M=1$. Letting $M=1$ can only give a weaker lower bound. However, we will show that the lower bound is tight in Theorem \ref{thm:1Dub}. 
Then for $(1-A_{\text{void}})\le  \epsilon(n)$
to hold, we need either (i) for $\epsilon(n) < r_s(n)$, at least one sensor in both intervals $\left[x - r_s(n), x-r_s(n)+\frac{\epsilon(n)}{2}\right]$ and $\left[x+r_s(n)-\frac{\epsilon(n)}{2}, x+ r_s(n)\right]$, where sensors in both the intervals give reading $1$, we call this event $A$,
or (ii) at least one sensor in the following four intervals, 
$\left[x-r_s(n), x\right], \left[x, x+r_s(n)\right],  \left[x-r_s(n)-\frac{\epsilon(n)}{2}, x-r_s(n)\right],  \left[x+r_s(n), x+r_s(n)+ \frac{\epsilon(n)}{2}\right]$, we call this event $B$,
where the sensors in the first two intervals give readings $1$, while the sensors in the last two intervals give readings $0$. We illustrate the events $A$ and $B$ in Fig.~\ref{fig:lb}. Since the  transmitter can be arbitrarily located, $x$ can be anywhere in $\yell$.  Thus, essentially, we need all intervals of length $r_s(n)$ and $\epsilon(n)/2$ to have at least one sensor.  
\begin{figure*}[t]
\centering
\includegraphics[width=4.5in]{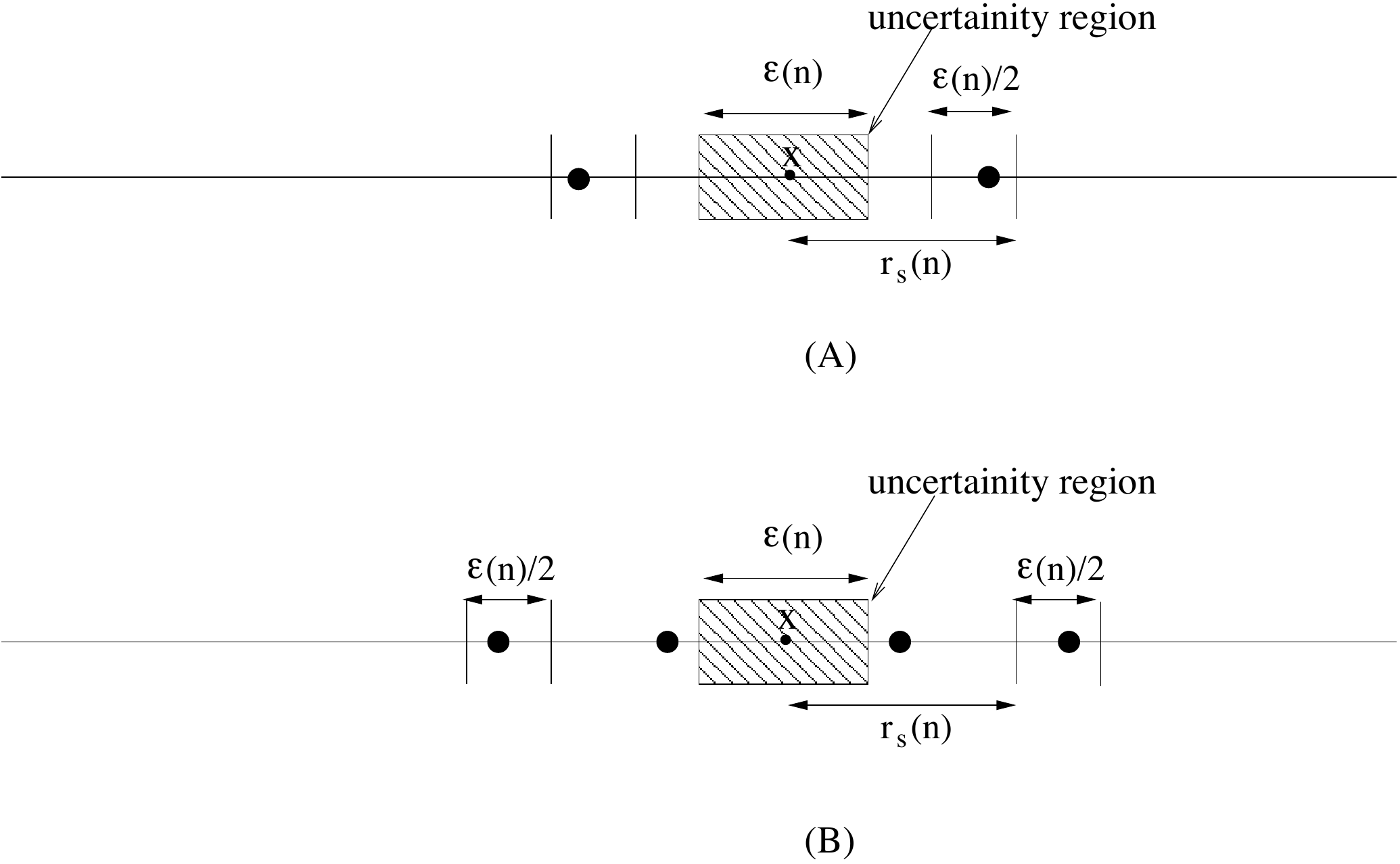}
\caption{Depiction of the uncertainty in transmitter location for the lower bound.}
\label{fig:lb}
\end{figure*}

To find the probability that all intervals of length $r_s(n)$ and $\epsilon(n)/2$ have at least one sensor, we use Lemma \ref{lem:1-cov}. Note that the setting in this Theorem is identical to Lemma \ref{lem:1-cov}, where $n$ sensors (in place of transmitters) with radio range $r_s(n)$  are thrown randomly in $\yell$.
From Lemma \ref{lem:1-cov}, we know that if $r_s(n)$ is less than order $\frac{\log n}{n}$, then  $\lim_{n\rightarrow \infty}P(\text{each point in $\yell$ is covered} ) < c_2$, where $c_2 <1$ is a constant. If any point on $\yell$ is not covered, then surely the interval of length $2 r_s(n)$ around it has no sensor. 
Hence, if $2r_s(n)$ is less than order $\frac{\log n}{n}$, then from Lemma \ref{lem:1-cov} there exists an interval of width $2r_s(n)$ that does not have any sensor with probability greater than $1-c_2$. Similarly, if $\epsilon(n)$ is less than order $\frac{\log n}{n}$, there exists an interval of width $\epsilon(n)$ that does not have any sensor with probability greater than $1-c_2$ from Lemma \ref{lem:1-cov}, thereby violating the conditions for having $(1-A_{\text{void}})\le  \epsilon(n)$. Thus, if $\lim_{n\rightarrow \infty} \frac{\epsilon(n)}{\frac{\log n}{n}} \rightarrow 0$ or $\lim_{n\rightarrow \infty} \frac{r_s(n)}{\frac{\log n}{n}} \rightarrow 0$, then $P\left((1-A_{\text{void}}) \le \epsilon(n)\right)  <1$.
\end{proof}

The result for the 2-dimensional region is as follows.
\begin{thm}\label{thm:1Dlb2D} For a 2-dimensional region, if $\lim_{n\rightarrow \infty} \frac{\epsilon(n)}{\sqrt{\frac{\log n}{n}}} \rightarrow 0$ or $\lim_{n\rightarrow \infty} \frac{r_s(n)}{\sqrt{\frac{\log n}{n}}} \rightarrow 0$, then $P\left((1-A_{\text{void}}) \le \epsilon(n)\right)  < 1$.
\end{thm}
\begin{proof} Similar to the proof of Theorem \ref{thm:1Dlb}, using the $2$-dimensional part of Lemma~\ref{lem:1-cov}.
\end{proof}

Next, we show that if $r_s(n) = \Theta\left(\frac{\log n}{n}\right)$ and  $\epsilon(n) = \Theta\left(\frac{\log n}{n}\right)$, then the lower bound on the whitespace recovery error obtained in Theorem \ref{thm:1Dlb} is tight. For the proof,  we will need the following Chernoff bound result.

\begin{lemma}\label{lem:chernoff} Let $X_1, X_2, \dots$ be identical and independently distributed Bernoulli random variables with mean $\mu = \mathbb{E}\{X_i\}$, and let $X = \sum_{i=1}^nX_i$. Then for $0<\delta < 1$, we have that 
$P(X < (1-\delta)\mu) \le \exp\left(-\frac{n \delta^2 \mu}{2}\right)$.
\end{lemma}

\begin{thm}\label{thm:1Dub}  If $r_s(n) = \Theta\left(\frac{\log n}{n}\right)$, then for $\epsilon(n) = \Theta\left(\frac{\log n}{n}\right)$, $P\left( 1-A_{\text{void}} \le \epsilon(n)\right)  \rightarrow 1$.
\end{thm}
\begin{proof} Let $r_s(n) =  \left(\frac{c \log n}{n}\right)$, where $c>1$ is a constant. Divide  $\yell$ into smaller non-overlapping intervals of length $\left(\frac{c \log n}{n}\right)$, and index these segments as $z_1$ to $z_{\left(\frac{n}{c \log n}\right)}$. Let the number of sensors lying in $z_k$ be $N_k$. From the Chernoff bound in Lemma \ref{lem:chernoff}, $P( N_k < \frac{c \log n}{2}) \le n^{-c/8}$, and taking the union bound, $P(N_k < \frac{c \log n}{2} \ \text{for any}\  k)\  < c_4n^{1-c/8}$, where $c_4$ is a constant. Thus, for large enough $c$, with high probability, each small interval $z_k$  contains at least  $\frac{c \log n}{2} $ sensors. 

Since there are $M$  transmitters, at the maximum, only $M$ smaller intervals among $z_1 \dots z_{\left(\frac{n}{c \log n}\right)} $ contain any  transmitter. Since the radio range $r_s= \left(\frac{c\log n}{n}\right)$,  a  transmitter lying in interval $z_k$ can only influence readings of sensors lying in $3$ adjacent intervals $z_{k-1}, z_{k},$ and, $z_{k+1}$. Therefore, there are at least $ \left(\frac{n}{c \log n}\right) - 3M$ intervals among $z_1, \dots, z_{\left(\frac{n}{c \log n}\right)} $ in which all sensor readings are $0$. Thus, an area of width $ \left(\left(\frac{n}{c \log n}\right) - 3M\right) \left(\frac{c \log n}{n}\right)$ contains no  transmitter, i.e. $A_{\text{void}} > \left(\left(\frac{n}{c \log n}\right) - 3M\right) \left(\frac{c \log n}{n}\right)=1-3M \left(\frac{c \log n}{n}\right)$. Therefore, with high probability, $(1-A_{\text{void}}) \le 3M \left(\frac{c \log n}{n}\right)$, proving the Theorem.
%
\end{proof}
\begin{rem}
The result for the 2-dimensional region follows similarly, with $r_s(n) = \Theta\left(\sqrt{\frac{\log n}{n}}\right)$ and  $\epsilon(n) = \Theta\left(\sqrt{\frac{\log n}{n}}\right)$. Further, all the results in the sequel are valid for 2-dimensional regions also, but we omit the formal statements to avoid repetition. 
\end{rem}
%

In this section, we have shown that asymptotically in $n$, the optimal radio range scales as $\frac{\log n}{n}$, and the corresponding optimal whitespace recovery error also scales as $\frac{\log n}{n}$. 
For finding the lower bound, we leveraged the results on the coverage problem. Then, we used a Chernoff bound result to show that if radio range is order $\frac{\log n}{n}$ then each interval of width $\frac{\log n}{n}$ contains $\log n$ sensors with high probability, and, hence, we can get a whitespace recovery accuracy of order $\frac{\log n}{n}$ for large enough $n$. In this section, we assumed that the sensor readings were error-free. Surprisingly, the above results hold even when the sensors are unreliable, as we show in the following section. 


\section{Unreliable Sensors} \label{sec:unreliable}

In this section, we assume that sensors are unreliable, and make an error in reading with probability 
$p < \frac{1}{2}$ independently of all other sensors. Thus, a sensor reading is $1$ even if there is no transmitter within a range of $r_s$ around it, or a sensor reading is $0$ even if there is a transmitter within a range of $r_s$ around it, and both events happen with probability $p$. In reality, the sensor errors could be a function of the distance from the transmitter, both in terms of missing a transmitter or identifying one when it is not present within the sensing radius $r_s$. Let the upper bound on the miss probability or false alarm of any sensor be $p$ over the sensing radius $r_s$. Then, our model, where each sensor makes an error with probability $p$, essentially takes care of the worst case scenario, while simultaneously simplifying the analysis.

As before, we are interested in finding the minimum error $\epsilon(n)$ and radio range $r_s(n)$ that solves the optimization problem 
\begin{equation}
\min_{r_s(n), \epsilon(n)}\lim_{n\rightarrow \infty} P\left((1-A_{\text{void}}) \le \epsilon(n)\right) =1.
\end{equation}

The following Theorem is the analog of Theorem \ref{thm:1Dlb}, with unreliable sensors. Its proof follows simply because the lower bound with unreliable sensors cannot be better than the lower bound with reliable sensors derived in Theorem \ref{thm:1Dlb}.


\begin{thm}\label{thm:1Dlbunreliable} 
When sensor measurements are in error with probability of error $p < \frac{1}{2}$, and for the whitespace recovery problem in a one-dimensional unit-length region, if $\lim_{n\rightarrow \infty} \frac{\epsilon(n)}{\frac{\log n}{n}} \rightarrow 0$ or $\lim_{n\rightarrow \infty} \frac{r_s(n)}{\frac{\log n}{n}} \rightarrow 0$, then $P\left((1-A_{\text{void}}) \le \epsilon(n)\right)  < 1$.
\end{thm}

Next, we show that the lower bound in Theorem \ref{thm:1Dlbunreliable} is also achievable. The proof is constructive, and follows by proposing a reconstruction strategy and analyzing its error performance.
\begin{thm}\label{thm:1Dubunreliable} When sensor measurements are in error with probability of error $p < \frac{1}{2}$, if $r_s(n) = \Theta\left(\frac{\log n}{n}\right)$, then for $\epsilon(n) = \Theta\left(\frac{\log n}{n}\right)$, $P\left( 1-A_{\text{void}} \le \epsilon(n)\right)  \rightarrow 1$ asymptotically in~$n$.
\end{thm}
\begin{proof} 
As in the proof of Theorem \ref{thm:1Dub}, let $r_s(n) =  \left(\frac{c \log n}{n}\right)$, where $c>1$ is a constant. Divide the segment $\yell$ into smaller intervals of length $\left(\frac{c \log n}{n}\right)$, and index these segments as $z_1$ to $z_{\left(\frac{n}{c \log n}\right)}$. Let the number of sensors lying in $z_k$ be $N_k$. From the Chernoff bound, $P( N_k < \frac{c \log n}{2}) \le n^{-c/8}$, and taking the union bound, $P(N_k < \frac{c \log n}{2} \ \text{for any}\  k)\  < c_4n^{1-c/8}$, where $c_4$ is a constant. Thus, with high probability, for large $c$, each smaller interval $z_k$  contains at least  $\frac{c \log n}{2} $ sensors. 

As before, since there are only $M$  transmitters, at the maximum only $M$ smaller intervals among $z_1 \dots z_{\left(\frac{n}{c \log n}\right)} $ contain any  transmitter. Since a  transmitter lying in $z_k$ can only influence readings of sensors lying in $3$ adjacent intervals $z_{k-1}, z_{k}$ and $z_{k+1}$. Therefore, in reality at least $ \left(\frac{n}{c \log n}\right) - 3M$ intervals among $z_1 \dots z_{\left(\frac{n}{c \log n}\right)} $ should have all sensor readings as $0$. However, because of errors in sensor readings, some of the sensors in these intervals have readings $1$ instead of $0$. 
To resolve this problem, we use the majority rule to decide whether an interval $z_k$ contains a  transmitter or not. Thus, a  transmitter is declared to be {\it present} in an interval $z_k$, if the number of sensors with reading $1$ are more than the number of sensors with reading $0$, and a  transmitter is declared to be {\it absent} in an interval $z_k$ otherwise. 

With this decision rule, $A_{\text{void}} = \cup_{\text{Maj}(z_k)=0} \ z_k$, where the function $\text{Maj}(z_k)$ equals $1$ if $z_k$ has a larger number of sensors with a reading of $1$ than with a reading of $0$, and equals $0$ otherwise. Therefore, the probability of interest in this case is that of missed detection, $P_{md}$, which is the probability that the majority of sensor readings in $z_k$ is $0$, given that there is a  transmitter in an interval $z_k$. Recall that each sensor makes an error with probability $p$ independently of all other sensors.
From the Chernoff bound, we know that in each interval $z_k$ there are at least $\frac{c \log n}{2} $ sensors with high probability, for large enough $c$. Let $N_k$ denote the number of sensors in $z_k$. 
Without loss of generality, we assume that $N_k$ is odd, as otherwise, we can consider one less sensor for making decisions. 
Then $P_{md} = \sum_{k=\frac{N_k+1}{2}}^{N_k} \binom{N_k}{k}p^{k}(1-p)^{N_k-k}$. 
  From an upper bound known in coding theory for repetition codes \cite{BookBlahut}, $\sum_{k=\frac{N_k+1}{2}}^{N_k} \binom{N_k}{k}p^{k}(1-p)^{N_k-k} \le \left(2 \sqrt{p(1-p)}\right)^{N_k}$.  Hence, for $p< \frac{1}{2}$, $P_{md} \le a^{N_k}$, where $a<1$, and $N_k$ is of the order ${ \log n}$ with high probability. 
Thus, the probability of missed detection $P_{md}$ decreases exponentially with $\frac{c \log n}{2}$. Since there are at the maximum $\left(\frac{n}{c \log n}\right)$ intervals in $\yell$, using the union bound, the probability that there is a missed detection in any one of the $\left(\frac{n}{c \log n}\right)$ intervals is $\left(\frac{n}{c \log n}\right) {a}^{\frac{c \log n}{2}}$, where $a<1$. Thus, for $c$ large enough,  $\left(\frac{n}{c \log n}\right) {a}^{\frac{c \log n}{2}} \rightarrow 0$ with high probability as $n\rightarrow \infty$. 

Using a similar analysis, we can show that the probability of false alarm in any interval $z_k$, i.e., the probability that the majority of sensor readings in $z_k$ is $1$, given that there is no  transmitter in an interval $z_k$,  goes to zero as $n$ goes to infinity. 

Thus, with high probability, we have that $ \left(\left(\frac{n}{c \log n}\right) - 3M\right)$ intervals  not containing any  transmitter have their majority of reading equal to $0$.  Hence, $A_{\text{void}} > \left(\left(\frac{n}{c \log n}\right) - 3M\right) \left(\frac{c \log n}{n}\right)=1-3M \left(\frac{c \log n}{n}\right)$ with high probability, proving the Theorem.
\end{proof}

In this section, we considered the case when each sensor makes an error with probability $p$ in the detection of any transmitter within its radio range. The lower bound on the whitespace recovery error is the same as in the case of reliable sensors, since the error with unreliable sensors cannot be better than that with reliable sensors. For finding the matching upper bound on the whitespace recovery error, we let the radio range be of order $\frac{\log n}{n}$, so that each interval of width $\frac{\log n}{n}$ contained more or less $\log n$ sensors with high probability. Then, for each interval of width $\frac{\log n}{n}$, we proposed a  majority rule for declaring the presence or absence of transmitter in that interval. Since there are a large number of sensors  (roughly $\log n$) in each interval, if $p < 1/2$, it followed from a repetition code argument that the probabilities of false alarm and missed detection go to zero for large $n$. Thus, we showed that, if the radio range is such that there are enough number of sensors in each small interval, asymptotically, the lack of reliability of the sensors has  no effect on the whitespace recovery error.


Note that, in the whitespace identification problem discussed above, we used the locations of sensors that returned a $0$ measurement to find the available void space that is guaranteed to not contain any transmitter. 
In the next section, we consider the problem of estimating the number of transmitters and their locations, with the absolute error in locating the transmitters as the metric of interest. We show that when the number of transmitters is unknown, a localization error of $\frac{\log(n)}{n}$ can still be achieved in the large number of sensors regime, provided the transmitters are known to be at a minimum separation of the order $\frac{\log(n)}{n}$ from each other.

\section{Transmitter Localization} \label{sec:localization}
In Section \ref{sec:whitespace1}, we considered the problem of finding the whitespace  $\mathcal{A}_{\text{void}}$ that contains no transmitters using binary sensors randomly deployed in the area. In addition to finding the void space area, there are several related problems of interest, such as finding the {\it received power profile} that describes the received power from the transmitters at each point of the given area $\yell$, finding the number and locations of the transmitters, etc. Towards that end, in this section, we are interested in finding how many transmitters  are present and their locations on $\yell$ using binary readings from the $n$ sensors that are uniformly randomly distributed on $\yell$.

As in the previous section, each sensor is assumed to have sensing radius $r_s(n)$, and has two possible readings, $1$ if there is at least one transmitter at a distance of $r_s$ from it, or $0$ otherwise. For simplicity, we consider the case of reliable sensors in this section. Results with unreliable sensors follow similarly, as in Sec. \ref{sec:unreliable}. To estimate the number of transmitters and their locations, we note that each disjoint region containing sensors that returned the value $1$ contains at least one transmitter. 
Hence, we estimate the number of transmitters to be equal to the number of disjoint regions containing sensors that returned the value $1$, and we estimate the transmitter locations $\hat{x}_i$ to be the geometric centroid of each such region. Note that, any contiguous region containing sensors that returned the value of $1$ could potentially have more than $1$ transmitter.\footnote{In particular, if the region is of width greater than $2r_s$, then it must necessarily contain more than one transmitter.} This could lead to errors in estimating the number of transmitters and/or their locations, as there is no way of identifying the number of transmitters within regions containing sensors that measured a $1$. To overcome this, in this section, we assume that any two transmitters are at at least $\delta(n)>0$ distance apart. As we will see, under mild assumptions on $\delta(n)$, it is possible to correctly estimate the number of transmitters and their locations with high probability, asymptotically in $n$. 
 
Let there be $M$ transmitters on $\yell$, and the true location of transmitter $j$ be 
$x_j$. Let the estimate of $M$, the number of transmitters, be ${\hat M}$, and the 
estimate of the location of the $i^{th}$ transmitter using the $n$ sensor readings be ${\hat x}_i, i =1,2,\dots,{\hat M}$, as described above. For both the true locations and their estimates, we index the transmitters  from left to right on $\yell$, such that $x_1\le x_2 \le \dots \le x_M$ and ${\hat x}_1 \le {\hat x}_2 \le \dots \le {\hat x}_{{\hat M}}$. 
Then we define the error metric as $\sum_{i=1}^{\max \{M, {\hat M}\}} |x_i -{\hat x}_i|$, where by definition we have that for $M < {\hat M}$, $x_i =0, i = M+1, \dots, {\hat M}$, and for $M >  {\hat M}$, ${\hat x}_i =1, i = {\hat M}+1, \dots, M$. This metric penalizes a mismatch between the actual and estimated number of  transmitters, in addition to the error in localizing them.

We are interested in finding minimum error $\epsilon(n)$, transmitter separation $\delta(n)$ and radio range $r_s(n)$ that guarantees that 
\begin{equation} \label{eq:probestM}
 \min_{r_s(n), \delta(n), \epsilon(n)}\lim_{n\rightarrow \infty} P\left( \sum_{i=1}^{\max \{M, {\hat M}\}} |x_i -{\hat x}_i| \le \epsilon(n)\right) =1.
 \end{equation} 

The next Theorem characterizes the lower bounds on  $\epsilon(n)$, $r_s(n)$ and $\delta(n)$ required for high probability estimation of the number and locations of the transmitters.

\begin{thm}\label{thm:1Dlbloc} If $\lim_{n\rightarrow \infty} \frac{\epsilon(n)}{\frac{\log n}{n}} \rightarrow 0$, then  $P\left( \sum_{i=1}^{\max \{M, {\hat M}\}} |x_i -{\hat x}_i| \le \epsilon(n)\right)  < 1$. Similarly, if $\lim_{n\rightarrow \infty} \frac{r_s(n)}{\frac{\log n}{n}} \rightarrow 0$, or $\lim_{n\rightarrow \infty} \frac{\delta(n)}{\frac{\log n}{n}} \rightarrow 0$ then  $P\left( \sum_{i=1}^{\max \{M, {\hat M}\}} |x_i -{\hat x}_i| \le \epsilon(n)\right)  < 1$.
\end{thm}
\begin{proof} 
First, the bounds on $\epsilon(n)$ and $r_s(n)$ follow from Theorem \ref{thm:1Dlb} with $M=1$. For the lower bound on $\delta(n)$, consider $M=2$ transmitters at locations $x_1$ and $x_2$ with distance $|x_2-x_1| = \delta(n)$ between them. To be able to decide that two transmitters are present,  i) at least one sensor has to lie in between $x_1$ and $x_2$ with a reading of $0$, or 
ii) $r_s(n)$ has to be less than or equal to $\delta(n)$, since otherwise the sensors lying  outside of the interval $(x_1,x_2)$ cannot discern whether there are two transmitters or one, as both $x_1$ and $x_2$ will possibly be in their range $r_s$.



Since the two transmitters can be arbitrarily located on $\yell$, condition i) implies that each interval of length $\delta(n)$ on $\yell$ should contain at least one sensor. Similar to the proof of Theorem \ref{thm:1Dlb}, the probability that each interval of length $\delta(n)$ contains at least one sensor is upper bounded by a constant less than $1$ if $\lim_{n\rightarrow \infty}\frac{\delta(n)}{\frac{\log n}{n}} = 0$. We already know that, for 
$\lim_{n\rightarrow \infty} \frac{r_s(n)}{\frac{\log n}{n}} \rightarrow 0$,   $P\left( \sum_{i=1}^{\max \{M, {\hat M}\}} |x_i -{\hat x}_i| \le \epsilon(n)\right)  < 1$ . Thus, conditions i) and ii) together imply that for $\lim_{n\rightarrow \infty}\frac{\delta(n)}{\frac{\log n}{n}} = 0$, $P\left( \sum_{i=1}^{\max \{M, {\hat M}\}} |x_i -{\hat x}_i| \le \epsilon(n)\right)  < 1$. 
\end{proof}

Our next result shows that $\epsilon(n) = r_s(n) = \delta(n) = \Theta\left(\frac{\log n}{n}\right)$ is sufficient for estimating the number and location of transmitters with high probability, asymptotically in~$n$.

\begin{thm}\label{thm:1Dubloc} 
If $r_s(n) = \delta(n) = \epsilon(n) =\Theta\left(\frac{\log n}{n}\right)$, $P\left( \sum_{i=1}^{\max \{M, {\hat M}\}} |x_i -{\hat x}_i| \le \epsilon(n)\right)   \rightarrow 1$.
\end{thm}
\begin{proof} Let $r_s(n) = \frac{c\log n}{n}$, $c > 1$, and let the minimum transmitter separation $\delta(n) =  \frac{d \log n}{n}$, where $d >10c$. Divide the region $\yell$ into smaller intervals of length $\left(\frac{10c \log n}{n}\right)$, and index these segments as $z_1$ to $z_{\left(\frac{n}{10c \log n}\right)}$. Hence, each interval $z_k$ contains at most one transmitter.

%
Partition each $z_i$ into five equal parts of width $ \left(\frac{2c \log n}{n}\right)$, and index them with $P_{i,1},\ldots, P_{i,5}$. 
Let the number of sensors lying in $P_{i,j}$ be $N_{i,j}$. From the Chernoff bound, $P( N_{i,j} < c \log n) \le n^{-c/4}$, and taking the union bound, $P(N_{i,j} < c \log n \ \text{for any}\  i, j=1,\dots,5)\  < c_4n^{1-c/4}$, where $c_4$ is a constant. Thus, with high probability, each partition of each interval contains at least  $c \log n $ sensors for large enough~$c$.
\begin{figure}[t]
\centering
\includegraphics[width=4.5in]{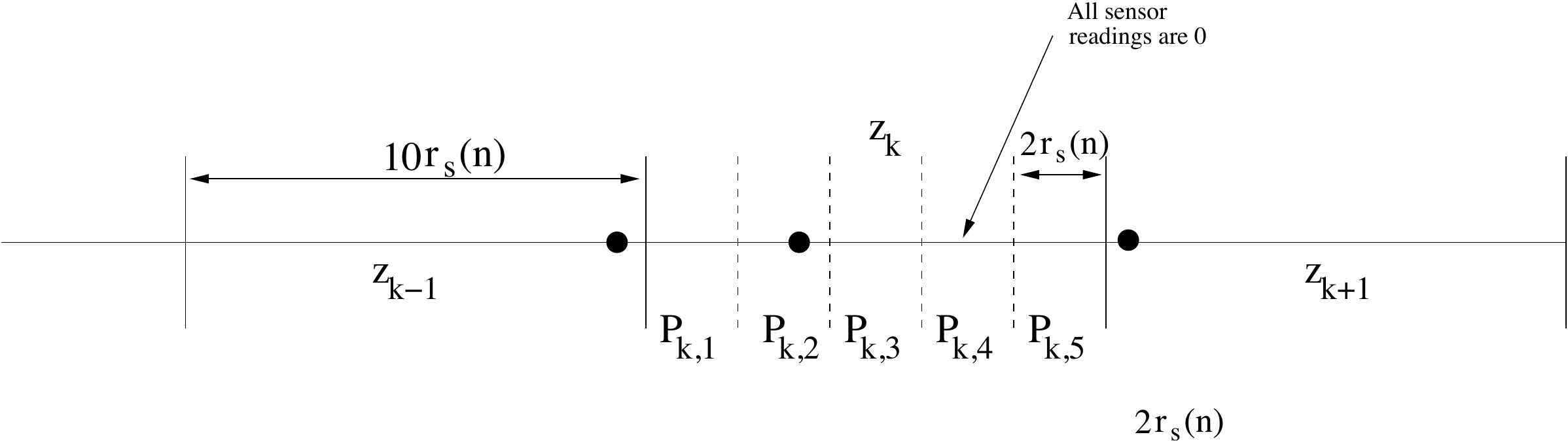}
\caption{Worst case placement of transmitters for estimating their locations.}
\label{fig:worsttxloc}
\end{figure}

Consider any interval $z_k$. If all the sensor readings in $z_k$ are zero, or if only the readings of the sensors lying in left half of $P_{k,1}$ or right half of $P_{k,5}$ are $1$, then no transmitter lies in $z_k$. Otherwise, we know that there is a  transmitter lying in $z_k$, say $x_i$.  Note that, it is hardest to detect the location of transmitter lying in $z_k$ if there are transmitters in both intervals $z_{k-1}$ and $z_{k+1}$, and they lie closest to the boundary of $z_{k}$, as shown in Fig. \ref{fig:worsttxloc}, where black dots represent the transmitters. 
Let  $x_i \in P_{k,j}$. 
Then, an interval $W_i$ of width $ \left(\frac{2c \log n}{n}\right)$ around $x_i$ contains at least $ \left(\frac{c \log n}{n}\right)$ sensors with high probability from the Chernoff bound, and all these sensors have reading $1$. In addition, irrespective of the index $j$ of partition $P_{k,j}$ to which $x_i$ belongs, there exists $\ellm$, $\ellm\in\{1,\dots,5\}$ for which all sensors lying in the partition $P_{k,\ellm}$ have a reading of $0$, since all the sensors lying in $P_{k,\ellm}$ are at a distance greater than the radio range $ \left(\frac{c \log n}{n}\right)$ from the transmitter $x_i$ in $P_{k,j}$. For example, in Fig. \ref{fig:worsttxloc}, all sensors lying in $P_{k,4}$ have their reading equal to $0$. 
Hence, using the readings from sensors in $W_i$ and $P_{k,\ellm}$, we can identify the location of the transmitter located inside it, and the uncertainty about the $i^{th}$ transmitter location is no more than two times the width of any partition $P_{k,m}$. This is equal to $ \left(\frac{4c \log n}{n}\right)$, and hence, $|{\hat x}_i - x_i| < \left(\frac{4c \log n}{n}\right)$. Since this is true for each transmitter $i$, ${\hat M} = M$, the total localization error $\sum_{i=1}^{\max \{M, {\hat M}\}} |x_i -{\hat x}_i| \le \sum_{i=1}^M\left(\frac{4c \log n}{n}\right) \le M\left(\frac{4c \log n}{n}\right)$ with high probability. This concludes the proof.
\end{proof}

In this section, we considered the problem of estimating both the number of transmitters as well as their locations using $n$ sensors making binary measurements. We showed that not knowing the number of transmitters does not significantly change the localization error as long as the minimum separation between any two transmitters is of order $\frac{\log n}{n}$. We first showed that if the minimum transmitter separation is less than order $\frac{\log n}{n}$, then 
the localization error probability cannot go to zero. Conversely, with the  minimum transmitter separation of order $\frac{\log n}{n}$, using the Chernoff bound  we showed that if the radio range is of order $\frac{\log n}{n}$, we can partition $\yell$ into small enough intervals so that with high probability, no interval contains more than one transmitter, while simultaneously ensuring that there are enough sensors in each interval for detection of the transmitter with high probability. The main message of this section is that  if minimum separation between transmitters scales as $\frac{\log n}{n}$, then transmitter localization problem is invariant to the knowledge of the number of transmitters. For a practical scenario where transmitters are geographically separated, minimum separation requirement for our results  is easily satisfied and hence the whitespace or received power profile can be detected efficiently.

\begin{rem} Another transmitter localization problem of interest is when the number of transmitters $M$ scales with the number of sensors $n$. 
Theorem~\ref{thm:1Dubloc} suggests that if $M(n)$ scales such that the minimum distance between any two transmitters scales no faster than order $\frac{\log n}{n}$, then a localization error of $M(n)\frac{\log n}{n}$ can be guaranteed with high probability. So, clearly, for $M(n) ={\cal O}\left(\sqrt{\frac{n}{\log n}}\right)$, where the minimum distance between any two transmitters scales no faster than $\frac{\log n}{n}$ \cite{Manjunath2012WS}, we have that the localization error scales as $\sqrt{\frac{n}{\log n}}\frac{\log n}{n}$, i.e., as  $\sqrt{\frac{\log n}{n}}$. Thus, our results also extend to the case where the number of transmitters scales with~$n$, under certain conditions.
\end{rem}

\section{Optimum Distribution of the Sensor Locations} \label{sec:distribution}
Thus far, we assumed that the transmitters are arbitrarily located on ${\cal L}$, and  obtained bounds on the minimum localization error and the optimal sensing radius in the worst case scenario. In some scenarios, it may be possible to obtain the spatial distribution of the transmitters on ${\cal L}$, either as side information from the primary network or from long-term statistics collected by the sensors. In this section, we consider the optimization of the spatial distribution of the sensor locations. We assume that the transmitters are distributed over $\yell = [0,1)$ with pdf $f_X(x)$, and seek to find the optimum sensor distribution $f_\lambda(x)$ over $\yell$ that minimizes the probability of missing a transmitter.  
Mathematically, we wish to solve
\begin{equation}\label{eq:objfn}
P_f = \min_{f_\lambda(x): \int_0^1 f_\lambda(x) \dee x = 1} \int_0^1 \left( 1 - 2 r_s f_\lambda(x) \right)^n f_X(x) \dee x.
\end{equation}
In the above, given that the location of transmitter is $x$, $2 r_s f_\lambda(x)$ is the probability (for small $r_s$) that there is a sensor in an area $2 r_s$ around it. Hence, $\left( 1 - 2 r_s f_\lambda(x) \right)^n$ represents the probability that none of the sensors lie within the sensing range $r_s$ of the transmitter.  By averaging over the distribution of $x$, $P_f$ captures the probability that all the sensors have reading $0$, and completely miss the transmitter at a random location in $\yell$. Using elementary results from variational calculus \cite{gelfand2000calculus}, the optimum $f_\lambda(x)$ must satisfy
\begin{equation}\label{eq:Lagrangian}
-n \left( 1 - 2 r_s f_\lambda(x)\right)^{n-1} f_X(x) 2 r_s + \mu = 0,
\end{equation}
where $\mu$ is a Lagrange multiplier factor, and is chosen such that $\int_0^1 f_\lambda(x) \dee x = 1$. This leads to 
\begin{equation} \label{eq:flambdaopt}
f_\lambda(x) = \left( 1 - \left(\frac{\mu}{2 n r_s f_X(x)} \right)^{\frac{1}{n-1}}\right) \frac{1}{2r_s}.
\end{equation}
In the above, $f_\lambda(x)$ is taken to be $0$ for $x$ such that $f_X(x) = 0$, and when the right hand side is negative. In some cases, the above reduces to intuitively satisfying results. For example, when $f_X(x)  = 1, 0 \le x < 1$, the above implies that $f_\lambda(x) = 1, 0 \le x < 1$, i.e., the optimum density is also uniform. On the other hand, when $n=1$, i.e., when only one sensor is deployed, $f_\lambda(x)$ drops out of \eqref{eq:Lagrangian}, and hence, interestingly, any point density that is continuous and nonzero on $[0, 1)$ performs equally well. 

The value of $\mu$ that ensures $\int_0^1 f_\lambda(x) \dee x = 1$ has to be obtained using numerical techniques. This is not difficult, since the right hand side in \eqref{eq:flambdaopt} is monotonically decreasing in $\mu$, taking the value $1/2r_s > 1$ when $\mu = 0$, and taking the value $0$ as $\mu$ gets large. Thus, any simple numerical technique such as the bisection method can be used to find the value of $\mu$.

Now, substituting the optimum $f_\lambda(x)$ into \eqref{eq:objfn} and simplifying, we get 
\begin{equation}\label{eq:objfnopt}
P_f^{(\text{opt})} = \frac{(1 - 2 r_s)^n}{\left[ \int_0^1 (f_X(x))^{-\frac{1}{n-1}} \dee x \right]^{n-1}}.
\end{equation}
We recognize the denominator as the $\ell_p$ norm of $1/f_X(x)$, with $p=1/(n-1)$ (which is in fact a quasi-norm). 
Note that, substituting the uniform point distribution for $f_\lambda(x) = 1, 0 \le x < 1$ in \eqref{eq:objfn} results in $P_f^{(\text{unif})}= (1-2r_s)^n$. Thus, the performance improvement from the optimized point density depends on the magnitude of the denominator in \eqref{eq:objfnopt}. 
For example, consider the case were $X$ has a triangular distribution: $f_X(x) = 4x$ for $0 \le x < 1/2$, and $= 4(1-x)$ for $1/2 \le  x < 1$. With some algebra, it can be shown that \eqref{eq:objfnopt} reduces to
\begin{equation}
P_f^{(\text{opt})} = \frac{2(1-2 r_s)^n}{\left( \frac{n-1}{n-2}\right)^{n-1}} \approx \frac{2(1-2 r_s)^n}{e \left( \frac{n-1}{n-2}\right)}.
\end{equation} 
Thus, the optimum point density does improve performance over the uniform point density, but both scale as $(1-2 r_s)^n$ with $n$. When $r_s = (\log n)/n$, for large $n$, $(1-2 r_s)^n \approx 1/n^2$, i.e., the probability of missing a transmitter is inversely proportional to~$n^2$.


\section{Simulation Results}\label{sec:sims}
We now present Monte Carlo simulation results to illustrate the analytical results developed in this paper. We consider $M$ transmitters and $n$ sensors deployed uniformly at random locations over $\yell = [0,1)$. Sensors return a $1$ if there is a  transmitter within the sensing radius $r_s$ around them, and return $0$ otherwise. We identify the whitespace as the total area spanned by the $2 r_s$ regions around sensors that return $0$. To estimate the number of transmitters and their locations, we first identify the occupied space as the union of the $2 r_s$ regions around sensors that return $1$. Then, for each contiguous occupied region of width smaller than $2 r_s$, we identify one transmitter at the center of the region. In contiguous occupied regions of width greater than $2 r_s$, we identify $\lfloor \text{width}/2 r_s \rfloor$ transmitters,  placed uniformly in the region. We compute the probability of the whitespace recovered exceeding $1-\epsilon(n)$, i.e., the objective function in \eqref{eq:prob}, with $\epsilon(n) = (\log n)/n$, and the probability of the transmitter localization error as in \eqref{eq:probestM}, by averaging over $10,000$ instantiations of transmitter and sensor deployments. 

Figure \ref{fig:wsrecovery} shows the probability of the whitespace recovered exceeding $1-\epsilon(n)$, i.e., the objective function in \eqref{eq:prob}, versus the number of sensors $n$, with $M = 1$ and $4$ transmitters and $\epsilon(n) = \log(n)/n$. We see that $r_s(n) = \log(n)/n$ outperforms the other scaling factors, which is in line with the result in Theorem~\ref{thm:1Dub}. In Fig.~\ref{fig:estnumtx}, we plot the probability that the sum absolute error in localizing the transmitters is $< \epsilon(n)$, given by \eqref{eq:probestM}. We set $\epsilon(n) = \log(n)/n$, and compare the performance of three different scalings for $r_s$: $\log(n)/n$, $(\log(n)/n)^2$, and $\sqrt{\log(n)/n}$, for $M=1$ and $M=4$ transmitters. We see that $\log(n)/n$ captures the optimal scaling  of the radio range with $n$, and it significantly outperforms the other scalings considered. Moreover, even at moderate or low values of $n$, scaling $r_s(n)$ at a rate that is higher or lower than $\log(n)/n$ results in a significant degradation in the performance. 

\begin{figure}[t]
\centering
\includegraphics[trim = 3cm 9cm 2cm 9cm, width=5in]{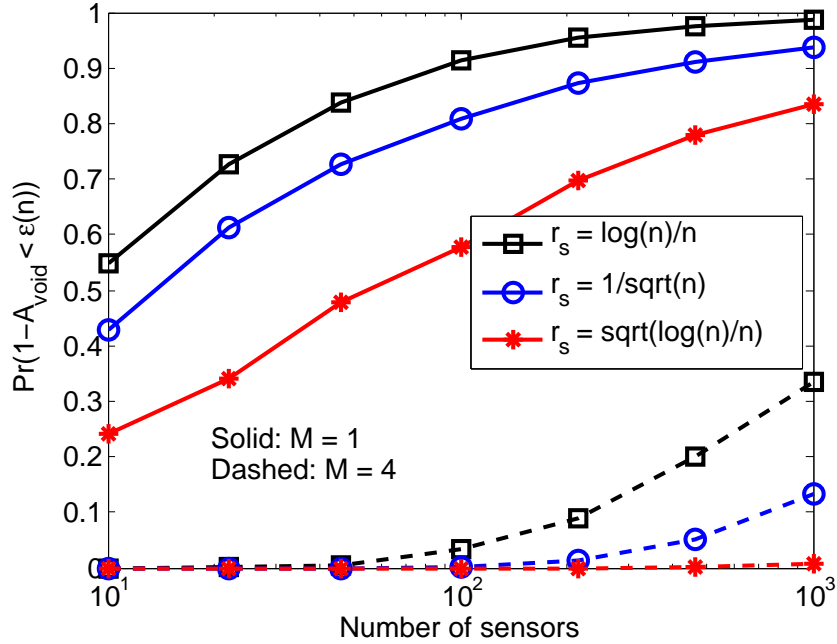}
\caption{Probability that the whitespace recovered is $> 1 - \epsilon(n)$, with $\epsilon(n) = \log(n)/n$.}
\label{fig:wsrecovery}
\end{figure}

\begin{figure}[t]
\centering
\includegraphics[trim = 3cm 9cm 2cm 9cm, width=5in]{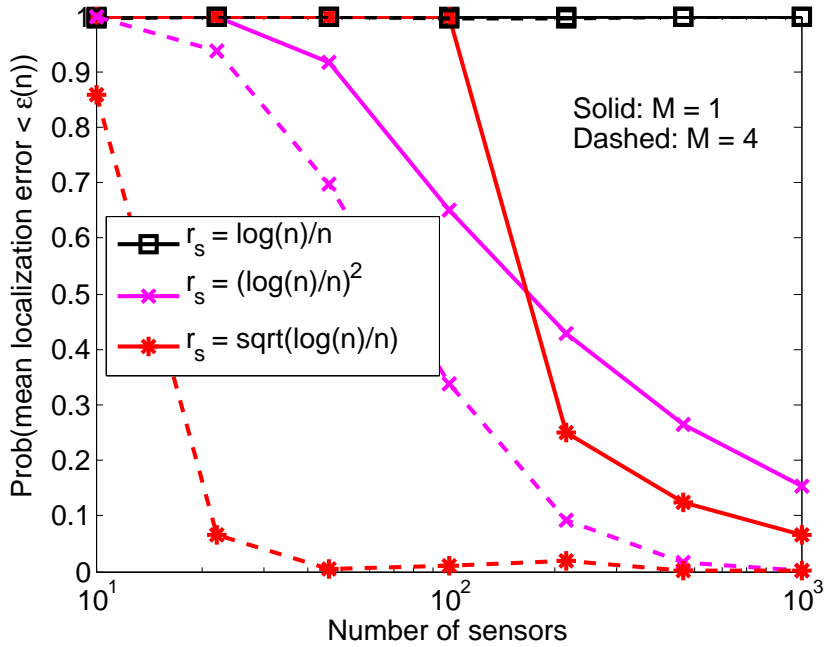}
\caption{Probability that the sum absolute error in localizing the transmitters is $< \epsilon(n)$, with $\epsilon(n) = \log(n)/n$.}
\label{fig:estnumtx}
\end{figure}

\begin{figure}[t]
\centering
\includegraphics[trim = 3cm 9cm 2cm 9cm, width=5in]{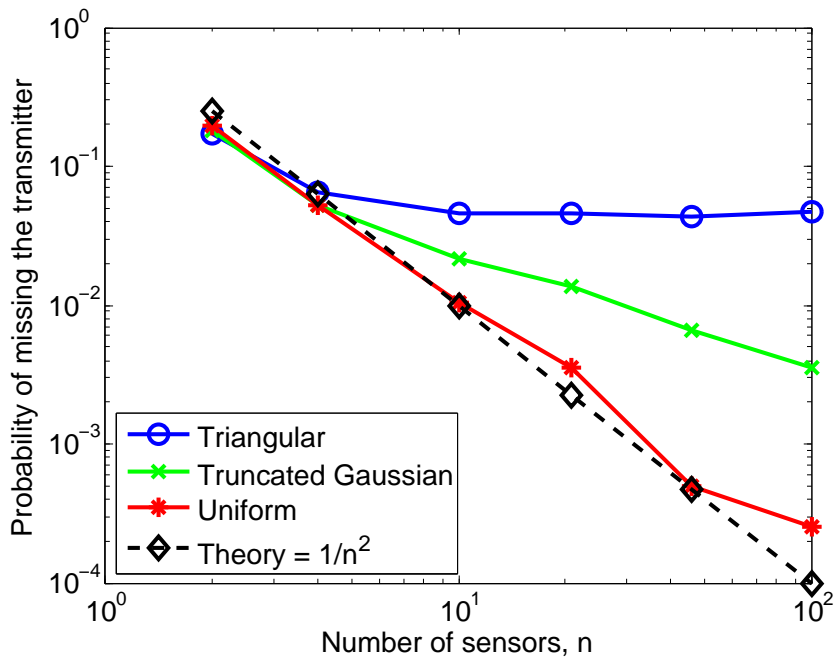}
\caption{Probability of missing a transmitter, when $n$ sensors with sensing radius $r_s(n) = \log(n)/n$ are deployed according to the triangular, truncated Gaussian and uniform distributions. }
\label{fig:prmisstx}
\end{figure}

Finally, Fig.~\ref{fig:prmisstx} shows the probability of missing a transmitter uniformly distributed on $[0,1)$, and $n$ sensors with sensing radius $r_s(n) = \log(n)/n$ are deployed according to the triangular, truncated Gaussian and uniform distributions. For the triangular distribution, we consider $f_\lambda(x) = 4x$ for $0 \le x < 1/2$, and $= 4(1-x)$ for $1/2 \le  x < 1$. For the truncated Gaussian distribution, we consider the Gaussian distribution with mean $0.5$ and standard deviation $0.25$, truncated to $[0, 1]$. We see that, as expected, the uniform distribution outperforms the other distributions, and its performance matches with the $P_f \approx 1/n^2$ result derived in Sec.~\ref{sec:distribution}. 

\section{Conclusions}\label{sec:conc}
In this paper, we studied the recovery of whitespace using $n$ sensors that are deployed at random locations within a given geographical area. We derived the limiting behavior of the recovered whitespace as a function of $n$ and  and the sensing radius $r_s$, and showed that both the whitespace recovery error (loss) and the radio range optimally scale as $\log(n)/n$ as $n$ gets large.  We also showed that, surprisingly, the radio range scaling of $\log(n)/n$ is optimal even with unreliable sensors. Using the sum absolute error in transmitter localization as the metric, we also analyzed the optimal scaling of the radio range that minimizes the localization error with high probability, as $n$ gets large. We also derived the corresponding optimal localization error, and showed that it scales as $\log(n)/n$ as well. Finally, we derived the optimal distribution of sensors that minimizes the probability of missing a transmitter, for a given distribution of the transmitters, and analyzed the behavior of the miss detection probability as $n$ is increased. Our results yielded useful insights into the number of sensors to be deployed and the  radio range for detecting transmitters that maximizes the recovered whitespace and accurately localizes the transmitters within the given geographical area. Future work could involve extending the sensor detection model to  account for temporal and spatial variations in the signal power due to shadowing, multipath fading, transmitter movement, etc.

\bibliographystyle{IEEEtran}
\bibliography{IEEEabrv,refs}

\begin{thebibliography}{10}
\providecommand{\url}[1]{#1}
\csname url@samestyle\endcsname
\providecommand{\newblock}{\relax}
\providecommand{\bibinfo}[2]{#2}
\providecommand{\BIBentrySTDinterwordspacing}{\spaceskip=0pt\relax}
\providecommand{\BIBentryALTinterwordstretchfactor}{4}
\providecommand{\BIBentryALTinterwordspacing}{\spaceskip=\fontdimen2\font plus
\BIBentryALTinterwordstretchfactor\fontdimen3\font minus
  \fontdimen4\font\relax}
\providecommand{\BIBforeignlanguage}[2]{{%
\expandafter\ifx\csname l@#1\endcsname\relax
\typeout{** WARNING: IEEEtran.bst: No hyphenation pattern has been}%
\typeout{** loaded for the language `#1'. Using the pattern for}%
\typeout{** the default language instead.}%
\else
\language=\csname l@#1\endcsname
\fi
#2}}
\providecommand{\BIBdecl}{\relax}
\BIBdecl

\bibitem{sandvig2006}
C.~Sandvig, ``Cartography of the electromagnetic spectrum: A review of wireless
  visualization and its consequences,'' in \emph{34th Telecomm. Policy Research
  Conf. (TPRC) on Commun., Information, and Internet Policy}, Arlington, VA,
  USA, 2006.

\bibitem{feki2008}
A.~Alaya-Feki, S.~Ben~Jemaa, B.~Sayrac, P.~Houze, and E.~Moulines, ``Informed
  spectrum usage in cognitive radio networks: Interference cartography,'' in
  \emph{Proc. IEEE PIMRC}, 2008, pp. 1--5.

\bibitem{Mateos2009}
G.~Mateos, J.~Bazerque, and G.~Giannakis, ``Spline-based spectrum cartography
  for cognitive radios,'' in \emph{43rd Asilomar Conf. on Signals, Systems and
  Computers}, 2009, pp. 1025--1029.

\bibitem{nasif2009measurements}
A.~O. Nasif and B.~L. Mark, ``{Measurement Clustering Criteria for Localization
  of Multiple Transmitters},'' in \emph{Proceedings of Conference on
  Information Systems and Sciences}, Baltimore, MD, USA, Mar. 2009, pp.
  341--345.

\bibitem{niu2006}
R.~Niu and P.~Varshney, ``Target location estimation in sensor networks with
  quantized data,'' \emph{IEEE Trans. on Sig. Proc.}, vol.~54, no.~12, pp.
  4519--4528, 2006.

\bibitem{shoari2010localization}
A.~Shoari and A.~Seyedi, ``Localization of an uncooperative target with binary
  observations,'' in \emph{Proc. IEEE SPAWC}, 2010, pp. 1--5.

\bibitem{YR2012}
\BIBentryALTinterwordspacing
Y.~Venugopalakrishna, C.~R. Murthy, and D.~N. Dutt, ``Multiple transmitter
  localization and communication footprint identification using energy
  measurements,'' \emph{Physical Communication}, 2012. [Online]. Available:
  \url{http://www.sciencedirect.com/science/article/pii/S1874490712000717}
\BIBentrySTDinterwordspacing

\bibitem{Madhow2009}
\BIBentryALTinterwordspacing
N.~Shrivastava, R.~Mudumbai, U.~Madhow, and S.~Suri, ``{Target tracking with
  binary proximity sensors},'' \emph{ACM Trans. Sen. Netw.}, vol.~5, no.~4, pp.
  1--33, 2009. [Online]. Available:
  \url{http://dx.doi.org/10.1145/1614379.1614382}
\BIBentrySTDinterwordspacing

\bibitem{Madhow2008WS}
R.~Mudumbai and U.~Madhow, ``Information theoretic bounds for sensor network
  localization,'' in \emph{Proc. IEEE ISIT}, 2008.

\bibitem{Aslam2003WS}
\BIBentryALTinterwordspacing
J.~Aslam, Z.~Butler, F.~Constantin, V.~Crespi, G.~Cybenko, and D.~Rus,
  ``Tracking a moving object with a binary sensor network,'' in \emph{Proc. 1st
  Int. Conf. on Embedded Networked Sensor Systems}.\hskip 1em plus 0.5em minus
  0.4em\relax New York, NY, USA: ACM, 2003, pp. 150--161. [Online]. Available:
  \url{http://doi.acm.org/10.1145/958491.958509}
\BIBentrySTDinterwordspacing

\bibitem{Manjunath2012WS}
K.~Santhana, A.~Kumar, D.~Manjunath, and B.~Dey, ``{Separability of a large
  number of targets using binary proximity sensors},'' \emph{To be submitted},
  2012.

\bibitem{Kim2005WS}
W.~Kim, K.~Mechitov, J.-Y. Choi, and S.~Ham, ``On target tracking with binary
  proximity sensors,'' in \emph{Proc. 4th Int. Symp. on Info. Proc. in Sens.
  Netw.}, 2005.

\bibitem{kumar04kcoverage}
\BIBentryALTinterwordspacing
S.~Kumar, T.~H. Lai, and J.~Balogh, ``On k-coverage in a mostly sleeping sensor
  network,'' in \emph{Proc. MobiCom}.\hskip 1em plus 0.5em minus 0.4em\relax
  New York, NY, USA: ACM, 2004, pp. 144--158. [Online]. Available:
  \url{http://doi.acm.org/10.1145/1023720.1023735}
\BIBentrySTDinterwordspacing

\bibitem{BookBlahut}
R.~Blahut, \emph{Algebraic Codes for Data Transmission}.\hskip 1em plus 0.5em
  minus 0.4em\relax Cambridge University Press, 2002.

\bibitem{gelfand2000calculus}
I.~Gelfand and S.~Fomin, \emph{Calculus of Variations}, ser. Dover Books on
  Mathematics.\hskip 1em plus 0.5em minus 0.4em\relax Dover Publications, 2000.

\end{thebibliography}


\end{document}